\documentclass[10pt,twocolumn]{IEEEtran-v17}
\pdfoutput=1

\usepackage{amssymb}
\usepackage{amsthm}
\usepackage{verbatim}
\usepackage{eepic}
\usepackage{epic}
\usepackage{graphicx}

\usepackage{amsmath}
\usepackage{amsfonts}

\newcommand{\bpsi}{\mbox{\boldmath{$\Psi$}}}
\newcommand{\bphi}{\mbox{\boldmath{$\Phi$}}}

\newcommand{\bpsil}{\mbox{\boldmath{$\psi$}}}
\newcommand{\bphil}{\mbox{\boldmath{$\phi$}}}

\newcommand{\bdel}{\mbox{\boldmath{$\delta$}}}

\newcommand{\deft}{{\stackrel{\triangle}{=}}}

\newcommand{\mub}{\mu_{\operatorname{B}}}

\newcommand{\argmax}{\operatornamewithlimits{arg\,max}}

\newcommand{\bbi}{{{\bf I}}}

\newcommand{\bbD}{{{\bf D}}}
\newcommand{\bbd}{{{\bf D}}}

\newcommand{\bbb}{{{\bf B}}}
\newcommand{\bbI}{{{\bf I}}}
\newcommand{\bbj}{{{\bf J}}}

\newcommand{\bbw}{{{\bf W}}}

\newcommand{\bbf}{{{\bf F}}}
\newcommand{\bbu}{{{\bf U}}}

\newcommand{\bbz}{{{\bf Z}}}
\newcommand{\bbm}{{{\bf M}}}
\newcommand{\bx}{{\bf x}}

\newcommand{\brv}{{\bf r}}

\newcommand{\by}{{\bf y}}
\newcommand{\ba}{{\bf a}}
\newcommand{\bb}{{\bf b}}

\newcommand{\bo}{{\bf 0}}
\newcommand{\bd}{{\bf d}}
\newcommand{\bc}{{\bf c}}

\newcommand{\I}{{\mathcal{I}}}

\newcommand{\R}{{\mathcal{R}}}

\newcommand{\CC}{{\mathbb{C}}}

\newcommand{\st}{\operatorname{s.t.} \,}

\newcommand{\bl}{\left(}
\newcommand{\br}{\right)}

\newcommand{\bba}{{\mathbf A}}

\newcommand{\obbd}{\overline{\bbd}}

\newcommand{\bv}{{\mathbf v}}

\newtheorem{theorem}{Theorem}
\newtheorem{lemma}{Lemma}
\newtheorem{proposition}{Proposition}

\title{Compressed Sensing of Block-Sparse Signals: Uncertainty Relations and Efficient Recovery}
\author{Yonina C. Eldar,~\IEEEmembership{Senior Member,~IEEE}, Patrick Kuppinger,~\IEEEmembership{Student Member,~IEEE}, and\\ Helmut B\"olcskei,~\IEEEmembership{Fellow,~IEEE}
\thanks{Y. Eldar is with the Department of Electrical Engineering, Technion, Haifa, Israel, Email: yonina@ee.technion.ac.il}
\thanks{H. B\"olcskei and P. Kuppinger are with the Communication Technology Laboratory, ETH Zurich, Zurich, Switzerland, Email: \{boelcskei,patricku\}@nari.ee.ethz.ch}
\thanks{This work was supported in part by the European Commission FP7 Network of Excellence in Wireless Communications NEWCOM++ and by the Israel Science Foundation.}
\thanks{This paper was presented in part at IEEE ICASSP 2009, Taipei, Taiwan, April~2009.}
}

\begin{document}

\maketitle

\begin{abstract}

We consider compressed sensing of block-sparse signals, i.e.,
sparse signals that have nonzero coefficients occurring in
clusters. An uncertainty relation for block-sparse signals is
derived, based on a block-coherence measure, which we introduce.
We then show that a block-version of the orthogonal matching
pursuit algorithm recovers block $k$-sparse signals in no more
than $k$ steps if the block-coherence is sufficiently small. The
same condition on block-coherence is shown to guarantee successful
recovery through a mixed $\ell_2/\ell_1$-optimization approach.
This complements previous recovery results for the block-sparse case which relied on small
block-restricted isometry constants. The significance of the
results presented in this paper lies in the fact that making explicit use of
block-sparsity can provably yield better reconstruction properties than
treating the signal as being sparse in the conventional sense,
thereby ignoring the additional structure in the problem.
\end{abstract}


\section{Introduction}

The framework of compressed sensing is concerned with the recovery of
an unknown vector from an underdetermined system of linear
equations \cite{Candes06,Donoho06}. The key property exploited for recovery
of the unknown data is the assumption of sparsity. 
More concretely, denoting by $\bx$ an unknown
vector that is observed through a measurement matrix $\bbd$ according to $\by=\bbd\bx$, it is
assumed that $\bx$ has only a few nonzero entries. A
fundamental observation is that if $\bbd$ is chosen properly and
$\bx$ is sufficiently sparse, then $\bx$ can be recovered from $\by=\bbd\bx$, irrespectively of the locations of
the nonzero entries of $\bx$, even if $\bbd$ has far fewer rows than columns. This result has
given rise to a multitude of different recovery algorithms which
can be proven to recover a sparse vector $\bx$ under a variety of
different conditions on $\bbd$
\cite{avdama97,efhajoti04,Tropp04,Candes06,CT05}.

Two widely studied recovery algorithms
are the basis pursuit (BP), or $\ell_1$-minimization approach
\cite{chdosa99,Candes06}, and the orthogonal matching pursuit (OMP) algorithm
\cite{Mallat93}. One of the main tools for the characterization of the recovery
abilities of BP is the restricted isometry property
(RIP) \cite{Candes06,C08}. Specifically, if the measurement
matrix $\bbd$ satisfies the RIP with appropriate restricted isometry constants, then
$\bx$ can be recovered by BP. Unfortunately, determining the RIP constants of a
given matrix is in general an NP-hard problem. A more simple and
convenient way to characterize recovery properties of a dictionary is via the coherence
measure \cite{DH01,ElBr02,Tropp04}. It was shown in \cite{Tropp04,DE03} that appropriate conditions on the coherence guarantee that both BP and OMP recover
the sparse vector $\bx$. The coherence also plays an important
role in uncertainty relations for sparse signals
\cite{DH01,ElBr02,E08a}.

In this paper, we consider compressed sensing of sparse signals that exhibit
additional structure in the form of the nonzero coefficients
occurring in clusters. Such signals are referred to as
block-sparse \cite{EM082,EldarMishali2009}. Our goal is to explicitly take
this block structure into account, both in terms of the recovery
algorithms and in terms of the measures that are used to
characterize their performance. The significance of the
results we obtain lies in the fact that making explicit use of
block-sparsity can provably yield better reconstruction properties than
treating the signal as being sparse in the conventional sense,
thereby ignoring the additional structure in the problem.

Block-sparsity arises naturally, e.g., when dealing with multi-band signals \cite{MishaliEldar2009,ME09,Landau67} or in 
measurements of gene expression levels \cite{PVMH08}. Another interesting
special case of the block-sparse model appears in the multiple measurement
vector (MMV) problem, which deals with the measurement of a set of vectors that
share a joint sparsity pattern \cite{Cotter,Chen,MishaliEldar2008a,EM082,ER09}.
Furthermore, it was shown in \cite{EM082,EldarMishali2009} that the block-sparsity model can be used to treat the problem of
sampling signals that lie in a union of subspaces \cite{LD08,Blumensath07,EM082,Eldar2009,E08a,MishaliEldar2009,ME09}.  

One approach to exploiting block-sparsity is by suitably
extending the BP method, resulting in a mixed $\ell_2/\ell_1$-norm recovery
algorithm \cite{EM082,Stoj08}. It was shown in \cite{EM082} that
if $\bbd$ has small block-restricted isometry constants, which
generalizes the conventional RIP notion, then the mixed norm
method is guaranteed to recover any block-sparse signal, irrespectively of
the locations of the nonzero blocks.
Furthermore, recovery will be robust in the presence of noise and modeling errors (i.e., when the vector is not exactly block-sparse). 
It was also established in \cite{EM082} that certain random matrices
satisfy the block RIP with overwhelming probability, and that this
probability is substantially larger than that of satisfying the
standard RIP. In \cite{richb08} extensions of the 
CoSaMP algorithm \cite{Tropp08} and of iterative hard
thresholding \cite{Blumensath08} to the model-based setting, which includes block-sparsity as a special case, 
are proposed and shown to exhibit provable recovery guarantees and robustness properties.

The focus of the present paper is on developing a parallel line of results by generalizing the notion
of coherence to the block setting. This can be seen as extending
the program laid out in \cite{Tropp04,DE03} to the block-sparse
case. Specifically, we define two separate
notions of coherence: coherence within a block, referred to as
sub-coherence and capturing local properties of the dictionary, and
block-coherence, describing global dictionary properties. We will
show that both coherence notions are necessary to characterize the essence of block-sparsity.
We present extensions of the BP, the matching pursuit (MP), and the OMP algorithms to
the block-sparse case and prove corresponding performance guarantees.

We point out that the term block-coherence was used previously in
\cite{Peotta2007} in the context of quantifying the recovery
performance of the MP algorithm in block-incoherent dictionaries.
Our definition pertains to block-versions of the MP and the OMP
algorithm and is different from that used in \cite{Peotta2007}.

We begin, in Section~\ref{sec:bs}, by introducing our definitions
of block-coherence and sub-coherence. In Section~\ref{sec:uc}, we
establish an uncertainty relation for block-sparse signals, and
show how the block-coherence measure defined previously occurs naturally in this
uncertainty relation. In Section~\ref{sec:er}, we introduce a block
version of the OMP algorithm, termed BOMP, and of the MP algorithm \cite{Mallat93}, termed BMP, and find a sufficient condition
on block-coherence that guarantees recovery of block $k$-sparse
signals through BOMP in no more than $k$ steps as well as exponential convergence of BMP.
The same condition on
block-coherence is shown to guarantee successful recovery through
the mixed $\ell_2/\ell_1$ optimization approach. The BOMP algorithm can be viewed as an extension of the subspace
OMP method for MMV systems \cite{ER09}.
The proofs of our main results are
contained in Section~\ref{sec:proofth2}. A discussion on the performance
improvements that can be obtained through exploiting block-sparsity is provided in Section~\ref{sec:discuss}.
Corresponding numerical results are reported in Section~\ref{sec:num}. We conclude in Section~\ref{sec:conclude}.

Throughout the paper, we denote vectors by boldface lowercase
letters, e.g., $\bx$, and matrices by boldface uppercase letters,
e.g., $\bba$. The identity matrix is written as $\bbi$ or $\bbi_d$
when the dimension is not clear from the context. For a given
matrix $\bba$, $\bba^{T}$, $\bba^{H}$, and $\mbox{Tr}(\bba)$
denote its transpose, conjugate transpose, and trace,
respectively, $\bba^\dagger$ is the pseudo inverse, $\R(\bba)$
denotes the range space of $\bba$, $\bba_{i,j}$ is the element in
the $i$th row and $j$th column of $\bba$, and $\ba_{\ell}$ stands
for the $\ell$th column of $\bba$. The $\ell$th element of a
vector $\bx$ is denoted by $x_{\ell}$.  The Euclidean norm of the
vector $\bx$ is $\|\bx\|_2=\sqrt{\bx^H\bx}$,
$\|\bx\|_1=\sum_{\ell} |x_{\ell}|$ is the $\ell_1$-norm,
$\|\bx\|_{\infty}=\max_{\ell} |x_{\ell}|$ is the
$\ell_{\infty}$-norm, and $\|\bx\|_{0}$ designates the number of
nonzero entries in ${\bf x}$. The Kronecker product of the
matrices $\bba$ and $\bbb$ is written as $\bba \otimes \bbb$. The
spectral norm of $\bba$ is denoted by
$\rho(\bba)=\lambda^{1/2}_{\max}(\bba^H\bba)$, where
$\lambda_{\max}(\bbb)$ is the largest eigenvalue of the
positive-semidefinite matrix $\bbb$.


\section{Block-Sparsity and Block-Coherence}
\label{sec:bs}
\subsection{Block-sparsity}

We consider the problem of representing a vector $\by \in \CC^L$
in a given dictionary $\bbd$ of size $L \times N$ with $L<N$, so
that
\begin{equation}
\label{eq:samples} \by=\bbd\bx
\end{equation}
for a coefficient vector $\bx \in \CC^N$. Since the system of
equations (\ref{eq:samples}) is underdetermined, there are, in general,
many possible choices of $\bx$ that satisfy (\ref{eq:samples}) for a given $\by$.
Therefore, further assumptions on $\bx$ are needed to guarantee uniqueness of the representation.
Here, we consider the case of sparse vectors $\bx$, i.e., $\bx$ has
only a few nonzero entries relative to its dimension. The
standard sparsity model considered in compressed sensing \cite{Candes06,Donoho06} assumes that $\bx$
has at most $k$ nonzero elements, which can appear anywhere in the vector.
As discussed in \cite{richb08,EM082,EldarMishali2009} there are practical scenarios that involve vectors $\bx$ with
nonzero entries appearing in blocks (or clusters)
rather than being arbitrarily spread throughout the
vector. Specific examples include signals that lie in unions of subspaces \cite{Blumensath07,LD08,EM082,Eldar2009},
and multi-band signals \cite{MishaliEldar2009,ME09,Landau67}.

The recovery of block-sparse vectors $\bx$ from measurements $\by=\bbd\bx$ is the
focus of this paper. To define
block-sparsity, we view $\bx$ as a concatenation of blocks ---assumed throughout the paper to be of
length $d$--- with $\bx[\ell]$ denoting the $\ell$th block,
i.e.,
\begin{equation}
\label{eq:xblock} \bx=[\underbrace{x_1 \,\, \ldots \,\,
x_d}_{\bx^{T}[1]} \,\,\underbrace{x_{d+1}\,\,\ldots\,\,
x_{2d}}_{\bx^{T}[2]}\,\, \ldots\,\, \underbrace{x_{N-d+1}\,\,\ldots
\,\,x_{N}}_{\bx^{T}[M]}]^T
\end{equation}
where $N=Md$. We furthermore assume that $L=Rd$ with $R$ integer.
Similarly to (\ref{eq:xblock}), we can represent $\bbd$ as a
concatenation of column-blocks $\bbd[\ell]$ of size $L \times d$:
\begin{equation}
\label{eq:dblock} \bbd=[\underbrace{\bd_1 \,\, \ldots \,\,
\bd_d}_{\bbd[1]} \,\,\underbrace{\bd_{d+1}\,\,\ldots \,\,
\bd_{2d}}_{\bbd[2]}\,\, \ldots\,\,
\underbrace{\bd_{N-d+1}\,\,\ldots \,\,\bd_{N}}_{\bbd[M]}].
\end{equation}

A vector $\bx \in \CC^N$ is called block $k$-sparse if $\bx[\ell]$
has nonzero Euclidean norm for at most $k$ indices $\ell$. 
When $d=1$, block-sparsity reduces to conventional sparsity as defined in \cite{Candes06,Donoho06}. Denoting
\begin{equation}
\label{eq:mixed} \|\bx\|_{2,0} = \sum_{\ell=1}^M
I(\|\bx[\ell]\|_2>0)
\end{equation}
with the indicator function $I(\cdot)$,
a block $k$-sparse vector $\bx$ is defined as a vector that satisfies
$\|\bx\|_{2,0} \leq k$. In the remainder of the paper conventional sparsity will be referred to simply as sparsity, in contrast to block-sparsity.

We are interested in providing conditions on the dictionary $\bf{D}$
ensuring that the block-sparse vector $\bx$ can be recovered
from measurements $\by$ of the form (\ref{eq:samples}) through
computationally efficient algorithms. Our approach is partly
based on \cite{Tropp04,ElBr02,DE03} (and the mathematical techniques used
therein) where equivalent results are provided for the sparse
case. The two algorithms investigated are BOMP and a mixed $\ell_2/\ell_1$-optimization program
(referred to as L-OPT \cite{EM082}). It was shown in \cite{EM082} that L-OPT
yields perfect recovery if the dictionary $\bbd$ satisfies appropriate
restricted isometry properties. The purpose of this paper is to provide
recovery conditions for BOMP and L-OPT based on a suitably defined measure of
block-coherence. We will see that block-coherence plays a role similar to
coherence in the case of conventional sparsity.



Before defining block-coherence, we note that in order to
have a unique block $k$-sparse $\bx$ satisfying (\ref{eq:samples}) it
is clear that we need $R\,>\,k$ and the columns within each block $\bbd[\ell],\ell=1,2,...,M$, need
to be linearly independent. More generally, we have the following
proposition taken from \cite{EM082}.
\begin{proposition}
\label{prop:inv} The representation (\ref{eq:samples}) is unique
if and only if $\bbd {\bf g} \neq \bo$ for every ${\bf g} \neq \bo$ that
is block $2k$-sparse.
\end{proposition}

From Proposition~\ref{prop:inv} the columns of $\bbd[\ell]$
are linearly independent for all $\ell$. Throughout the paper, we
assume that the dictionaries we consider satisfy the
condition of Proposition \ref{prop:inv}, and, furthermore, $\|\bd_{r}\|_2=1,\,r=1,2,...,N$.


\subsection{Block-coherence}

The coherence of a dictionary $\bbd$ measures the similarity
between basis elements, and is defined by \cite{DH01,ElBr02,Tropp04}
\begin{equation}
\label{eq:cc} \mu=\max_{\ell, r \neq \ell} |\bd_{\ell}^H\bd_r|.
\end{equation}
This definition was introduced in
\cite{Mallat93} to heuristically characterize the performance of
the MP algorithm, and was later shown to play a fundamental role in quantifying
recovery thresholds for the OMP algorithm and for BP \cite{Tropp04}.
The coherence $\mu$ furthermore occurs in $\ell_1$-uncertainty relations relevant
in the context of decomposing a vector into two orthonormal bases
\cite{DH01,ElBr02}. A definition of coherence for analog
signals, along with a corresponding uncertainty relation, is
provided in \cite{E08a}.

It is natural to seek a generalization of coherence to
the block-sparse setting with the resulting block-coherence measure
having the same operational significance as the coherence $\mu$ in the
sparse case. Below, we propose such a generalization, which is shown ---in Sections~\ref{sec:uc} and \ref{sec:er}---
to occur naturally in uncertainty relations and in recovery thresholds for the block-sparse
case.

We define the \textit{block-coherence}
of $\bbd$ as
\begin{equation}
\label{eq:bc} \mub=\max_{\ell, r \neq \ell}
\frac{1}{d}\rho(\bbm[\ell,r])
\end{equation}
with
\begin{equation}
\label{eq:bcm}\bbm[\ell,r]=\bbd^H[\ell]\bbd[r].
\end{equation}
Note that $\bbm[\ell,r]$ is the $(\ell, r)$th $d\,\times\,d$
block of the $N \times N$ matrix $\bbm=\bbd^H\bbd$. When $d=1$,
as expected, $\mub=\mu$. While $\mub$ quantifies global properties of the dictionary $\bbd$, local
properties are described by the \textit{sub-coherence} of $\bbD$, defined as
\begin{align}\label{eq:subcoherence}
 \nu = \max_\ell \max_{i,j\neq i}|\bd_{i}^H\bd_j|,\quad \bd_i,\bd_j\in\bbD[\ell].
\end{align}
We define $\nu=0$ for $d=1$. In addition, if the columns of
$\bbd[\ell]$ are orthonormal for each $\ell$, then $\nu=0$.


Since the columns of $\bbd$ have unit norm, the coherence $\mu$ in (\ref{eq:cc}) satisfies $\mu\,\in\,[0,1]$ and therefore, as a consequence of $\nu\in[0,\mu]$, we have
$\nu\,\in\,[0,1]$. The following proposition establishes the same limits for the block-coherence $\mub$, which explains the choice of normalization by
$1/d$ in the definition (\ref{eq:bc}).



In the remainder
of the paper conventional coherence will be referred to simply as
coherence, in contrast to block-coherence and sub-coherence.
\begin{proposition}
\label{prop:cb} The block-coherence $\mub$ satisfies $0 \leq \mub \leq \mu$.
\end{proposition}
\begin{IEEEproof} Since the spectral norm is non-negative, clearly $\mub \geq 0$. To prove that $\mub\leq\mu$, note that the entries of $\bbm[\ell,r]$ for $\ell\neq r$ have absolute value smaller than or equal to $\mu$. It then follows that
\begin{align}
        \mub & = \max_{\ell,r\neq \ell}\frac{1}{d} \rho(\bbm[\ell,r]) \nonumber\\
        & = \max_{\ell,r\neq \ell}\frac{1}{d} \sqrt{\lambda_{\max}(\bbm^H[\ell,r]\bbm[\ell,r])} \nonumber\\
        &\leq \max_{\ell,r\neq \ell}\frac{1}{d} \sqrt{\max_i\sum_{j=1}^d \big\vert(\bbm^H[\ell,r]\bbm[\ell,r])_{i,j} \big\vert}\label{eq:Gershgorin}\\
        &\leq \max_{\ell,r\neq \ell}\frac{1}{d} \sqrt{\max_i\sum_{j=1}^d d\mu^2} \nonumber \\
        & = \mu
\end{align}
where \eqref{eq:Gershgorin} is a consequence of Ger\v{s}gorin's disc theorem (\cite[Corollary 6.1.5]{Horn85}).
\end{IEEEproof}
From $\mu\,\le\,1$, with Proposition~\ref{prop:cb}, it now follows trivially that $\mub\,\le\,1$.

When the columns of $\bbD[\ell]$ are orthonormal for each $\ell$, we can further bound $\mub$.
\begin{proposition}
\label{prop:orthoblocks} If $\bbD$ consists of orthonormal blocks,
i.e., $\bbD^H[\ell]\bbD[\ell]=\bbi_d$ for all $\ell$, then
$\mub\leq1/d$.
\end{proposition}
\begin{IEEEproof} Using the submultiplicativity of the spectral norm, we have
\begin{align}
        \mub & = \max_{\ell,r\neq \ell}\frac{1}{d} \rho(\bbm[\ell,r]) \nonumber\\
        & =  \max_{\ell,r\neq \ell}\frac{1}{d} \rho(\bbD^H[\ell]\bbD[r]) \nonumber \\
        &\leq \max_{\ell,r\neq \ell}\frac{1}{d} \rho(\bbD^H[\ell])\rho(\bbD[r]) \nonumber \\
        & = \frac{1}{d}\label{eq:rhoDl}
\end{align}
where \eqref{eq:rhoDl} follows from $\bbD^H[\ell]\bbD[\ell]=\bbi_d,$ for all $\ell$, $\lambda_{\max}(\bbD^H[\ell]\bbD[\ell])=\lambda_{\max}(\bbD[\ell]\bbD^H[\ell])$, and
$\lambda_{\max}(\bbi_d)=1$ combined with the definition of the spectral norm.


\end{IEEEproof}


\section{Uncertainty Relation for Block-Sparse Signals}
\label{sec:uc}

We next show how the block-coherence $\mub$ defined above
naturally appears in an uncertainty relation for block-sparse
signals. This uncertainty relation generalizes the corresponding
result for the sparse case derived in \cite{DH01,ElBr02}.

Uncertainty relations for sparse signals are concerned with
representations of a vector $\bx \in \CC^L$ in two
different orthonormal bases for $\CC^L$: $\{\bphil_{\ell},1 \leq
\ell \leq L\}$ and $\{\bpsil_{\ell},1 \leq \ell \leq L\}$
\cite{DH01,ElBr02}. Any vector $\bx\,\in\,\CC^L$ can be expanded
uniquely in terms of each one of these bases according to:
\begin{equation}
\label{eq:xn} \bx=\sum_{\ell=1}^L a_\ell
\bphil_{\ell}=\sum_{\ell=1}^L b_\ell \bpsil_{\ell}.
\end{equation}
The uncertainty relation sets limits on the sparsity of the
decompositions (\ref{eq:xn}) for any $\bx \in \CC^L$.
Specifically, denoting $A=\|\ba\|_0$ and $B=\|\bb\|_0$, it is
shown in \cite{ElBr02} that
\begin{equation}
\label{eq:ucd} \frac{1}{2}\bl A+B \br \geq \sqrt{AB} \geq
\frac{1}{\mu(\bphi,\bpsi)}
\end{equation}
where $\mu(\bphi,\bpsi)$ is the coherence between $\bphi$ and
$\bpsi$, defined as
\begin{equation}
\label{eq:mud}
\mu(\bphi,\bpsi)=\max_{\ell,r}|\bphil^H_{\ell}\bpsil_r|.
\end{equation}
It is easily seen that for $\bbd$ consisting of the
orthonormal bases $\bphi$ and $\bpsi$, i.e., $\bbd=[\bphi\,\,\bpsi]$, we have $\mu(\bphi,\bpsi)=\mu$, where $\mu$
is as defined in (\ref{eq:cc}) and associated with ${\bf
D}=[\bphi\,\,\bpsi]$.

In \cite{DH01} it is shown that $1/\sqrt{L} \leq \mu(\bphi,\bpsi)
\leq 1$. The upper bound follows from the Cauchy-Schwarz
inequality and the fact that the basis elements have norm $1$. The
lower bound is obtained as follows: The matrix
$\bbm=\bphi^H\bpsi$ is unitary so that
$\sum_{\ell=1}^{L}\sum_{r=1}^{L}|\bphil^H_{\ell}\bpsil_r|^{2}=\mbox{Tr}(\bbm^{H}\bbm)=\mbox{Tr}(\bbi_{L})=L$.
Consequently, we have $L^{2}\max_{\ell,r}|\bphil^H_{\ell}\bpsil_r|^{2}\,\ge\,L$ which implies
$\mu(\bphi,\bpsi)\,\ge\,1/\sqrt{L}$. This lower bound can be
achieved, for example, by choosing the two orthonormal bases $\bphi$ and $\bpsi$ as the spike
(identity) and Fourier bases \cite{DH01}. With this choice, the
uncertainty relation (\ref{eq:ucd}) becomes
\begin{equation}
\label{eq:ucdf} A+B  \geq 2\sqrt{AB} \geq 2\sqrt{L}.
\end{equation}
 When $\sqrt{L}$ is an
integer, the relations in (\ref{eq:ucdf}) can all be satisfied
with equality by choosing $\bx$ as a Dirac comb $\bdel_{\sqrt{L}}$
with spacing $\sqrt{L}$, resulting in $\sqrt{L}$ nonzero
elements. This follows from the fact that the Fourier transform of
$\bdel_{\sqrt{L}}$ is also $\bdel_{\sqrt{L}}$.

We now develop an uncertainty relation for block-sparse
decompositions. Specifically, we
derive a result that is equivalent to (\ref{eq:ucd}) with $A$ and
$B$ replaced by block-sparsity levels as defined in (\ref{eq:mixed}), and $\mu(\bphi,\bpsi)$
replaced by the block-coherence between the orthonormal bases
considered, and defined below in (\ref{eq:bc2}).
\begin{theorem}
\label{thm:uncertainty} Let $\bphi,\bpsi$ be two unitary $L\,\times\,L$
matrices with $L\,\times\,d$ blocks $\{\bphi[\ell],\bpsi[\ell],1 \leq \ell \leq
R\}$ and let $\bx \in \CC^L$ satisfy
\begin{equation}
\bx=\sum_{\ell=1}^R \bphi[\ell]\ba[\ell]=\sum_{\ell=1}^R
\bpsi[\ell]\bb[\ell].
\end{equation}
Let $A=\|\ba\|_{2,0}$ and $B=\|\bb\|_{2,0}$. Then,
\begin{equation}
\label{eq:uca} \frac{1}{2}(A+B)\geq \sqrt{AB} \geq \frac{1}{d
\mub(\bphi,\bpsi)}
\end{equation}
where
\begin{equation}
\label{eq:bc2} \mub(\bphi,\bpsi)=\max_{\ell,r}
\frac{1}{d}\rho(\bphi^H[\ell]\bpsi[r]).
\end{equation}
\end{theorem}

Note that for $\bbd$ consisting of the orthonormal bases
$\bphi$ and $\bpsi$, i.e., $\bbd=[\bphi\,\,\bpsi]$, we have
$\mub(\bphi,\bpsi)=\mub$, where $\mub$ is as defined in
(\ref{eq:bc}) and associated with $\bbd=[\bphi\,\,\bpsi]$.

\begin{IEEEproof}
Without loss of generality, we assume that $\|\bx\|_2^2=1$. Then,
\begin{align}
\label{eq:norm} 1=\|\bx\|_2^2&= \sum_{\ell,r=1}^R
\ba^H[\ell]\bba[\ell,r]\bb[r]\\ & \leq \sum_{\ell,r=1}^R
|\ba^H[\ell]\bba[\ell,r]\bb[r]|
\end{align}
where we set $\bba[\ell,r]=\bphi^H[\ell]\bpsi[r]$. Now, from
the Cauchy-Schwarz inequality, for any $\ba,\bb$,
\begin{align}
\left|\ba^H\bba[\ell,r]\bb\right| &\leq
\|\bb\|_2\|\bba^H[\ell,r]\ba\|_2 \nonumber\\
&  \leq  \lambda^{1/2}_{\max}(\bba[\ell,r]\bba^H[\ell,r]) \|\bb\|_2\|\ba\|_2\nonumber\\
&\leq d\mub \|\bb\|_2 \|\ba\|_2
\end{align}
where, for brevity, we wrote $\mub=\mub(\bphi,\bpsi)$. Substituting
into (\ref{eq:norm}), we get
\begin{equation}
\label{eq:norm2} 1 \leq  d\mub
\sum_{r=1}^R\|\bb[r]\|_2\sum_{\ell=1}^R \|\ba[\ell]\|_2.
\end{equation}
Applying the Cauchy-Schwarz inequality yields
\begin{equation}
\label{eq:norm3} \sum_{r=1}^R\|\bb[r]\|_2 \leq \sqrt{B}
\bl\sum_{r=1}^R\|\bb[r]\|_2^2\br^{1/2}=\sqrt{B}
\end{equation}
where we used the fact that
$\sum_{r=1}^R\|\bb[r]\|_2^2=\|\bb\|_2^2=1$ since $\|\bx\|_2^2=1$
and $\bpsi$ is unitary. Similarly, we have that
$\sum_{r=1}^R\|\ba[r]\|_2 \leq \sqrt{A}$. Substituting into
(\ref{eq:norm2}) and using the inequality of
arithmetic and geometric means completes the proof.
\end{IEEEproof}

The bound provided by Theorem~\ref{thm:uncertainty} can be tighter
than that obtained by applying the conventional uncertainty
relation (\ref{eq:ucd}) to the block-sparse case.
This can be seen by using $\|\ba\|_{0} \leq  d \|\ba\|_{2,0}$ and
$\|\bb\|_{0} \leq  d \|\bb\|_{2,0}$ in (\ref{eq:ucd}) to obtain
\begin{equation}
\sqrt{\|\ba\|_{2,0}\|\bb\|_{2,0}} \geq \frac{1}{d
\mu}.
\end{equation}
Since $\mub\,\le\,\mu$, this bound may be looser than
(\ref{eq:uca}).

\subsection{Block-incoherent dictionaries}
\label{sec:uca}

As already noted, in the sparse case (i.e., $d=1$) for any two orthonormal bases ${\bf \Phi}$ and ${\bf \Psi}$, we have $\mu\,\ge\,1/\sqrt{L}$. We next
show that the block-coherence satisfies a similar inequality, namely $\mub\,\ge\,1/\sqrt{dL}$.
\begin{proposition}
\label{prop:lb} The block-coherence \eqref{eq:bc2} satisfies
$\mub\,\ge\,1/\sqrt{dL}$.
\end{proposition}
\begin{IEEEproof}
Let ${\bf \Phi}$ and ${\bf
\Psi}$ be two orthonormal bases for $\CC^L$ and let
$\bba=\bphi^H\bpsi$ with $\bba[\ell,r]$ denoting the $(\ell,
r)$th $d\,\times\, d$ block of ${\bf A}$. With $R=L/d$, we have
\begin{eqnarray}
\label{eq:bco2} R^2 \mub^2 & \geq & \sum_{\ell=1}^R \sum_{r=1}^R
\frac{1}{d^2}\,\lambda_{\max}(\bba^H[\ell,r]\bba[\ell,r]) \nonumber
\\ &\geq &\frac{1}{d^2}\,\lambda_{\max}\!\bl \sum_{\ell=1}^R
\sum_{r=1}^R \bba^H[\ell,r]\bba[\ell,r]\br.
\end{eqnarray}
Now, it holds that
\begin{equation}
\label{eq:bco3} \sum_{\ell=1}^R \sum_{r=1}^R
\bba^H[\ell,r]\bba[\ell,r]= \sum_{r=1}^R \bpsi^H[r]\bl
\sum_{\ell=1}^R \bphi[\ell]\bphi^H[\ell] \br \bpsi[r].
\end{equation}
\sloppy Since $\bphi$ is a square matrix consisting of orthonormal columns, we have
$\sum_{\ell=1}^{R} \bphi[\ell]\bphi^H[\ell]=\bphi\bphi^H=\bbi_L$.
Furthermore, since $\bpsi[r]$ consists of orthonormal columns, for
each $r$, we have $\bpsi^H[r]\bpsi[r]= \bbi_d$. Therefore,
(\ref{eq:bco2}) becomes
\begin{equation}
\label{eq:bco4}   \mub^2 \geq \frac{1}{d^2 R}=\frac{1}{dL}
\end{equation}
which concludes the proof.
\end{IEEEproof}

We now construct a pair of bases that achieves the lower bound in (\ref{eq:bco4}) and
therefore has the smallest possible block-coherence. Let $\bbf$
be the DFT matrix of size $R=L/d$ with $\bbf_{\ell, r}=
(1/\sqrt{R})\exp(j 2\pi \ell r/R)$. Define $\bphi=\bbi_L$ and
\begin{equation}
\label{eq:kron} \bpsi=\bbf \otimes \bbu_d
\end{equation}
where $\bbu_d$ is an arbitrary $d\,\times\,d$ unitary matrix. For this
choice, \sloppy $\bphi^H[\ell]\bpsi[r]=\bbf_{\ell, r} \bbu_d$.
Since $\rho(\bbu_d)=1$ and $|\bbf_{\ell, r}|=1/\sqrt{R}$, we get
\begin{equation}
\label{eq:mubl} \mub=\frac{1}{d\sqrt{R}}=\frac{1}{\sqrt{dL}}.
\end{equation}
When $d=1$, this basis pair reduces to the spike-Fourier pair
which is well known to be maximally incoherent \cite{DH01}.

When $\mub$ satisfies (\ref{eq:mubl}) the uncertainty relation
becomes
\begin{equation}
\label{eq:ucdfb} A+B\geq 2 \sqrt{AB}\geq 2\sqrt{R}.
\end{equation}
If $\sqrt{R}$ is integer, the inequalities in (\ref{eq:ucdfb}) are met with equality for
the signal $\bx=\bdel_{\sqrt{R}} \otimes \bc$
where $\bc$ is an arbitrary nonzero length-$d$ vector. 
Indeed, in this case, the representation of $\bx$ in the spike
basis requires $\sqrt{R}$ blocks (of size $d$), so that
$\|\ba\|_{2,0}=\sqrt{R}$. The representation of $\bx$ in the basis $\bpsi$ in (\ref{eq:kron})
is obtained as
\begin{equation}
\bb=(\bbf^H \otimes \bbu_d^H)(\bdel_{\sqrt{R}} \otimes
\bc)=\bdel_{\sqrt{R}} \otimes \bbu_d^H\bc
\end{equation}
where we used the fact that the Fourier transform of
$\bdel_{\sqrt{R}}$ is also $\bdel_{\sqrt{R}}$. Therefore, $\bb$
has $\sqrt{R}$ nonzero blocks so that $\|\bb\|_{2,0}=\sqrt{R}$ and hence $A=B=\sqrt{R}$, which implies
that all inequalities in (\ref{eq:ucdfb}) are met with equality.

\section{Efficient Recovery Algorithms}
\label{sec:er}

We now give operational meaning to block-coherence by showing that
if it is small enough, then a block-sparse signal $\bx$ can be
recovered from the measurements ${\bf y}=\bbd \bx$ using computationally
efficient algorithms. We consider two different recovery methods, namely
the mixed $\ell_2/\ell_1$-optimization program (L-OPT) proposed in
\cite{EM082}:
\begin{eqnarray}
\label{eq:l1} \min_{\bx}  \sum_{\ell=1}^M \|\bx[\ell]\|_2 \quad
\st  \by=\bbd\bx
\end{eqnarray}
and an extension of the OMP algorithm \cite{Mallat93} to the
block-sparse case described below and termed block-OMP (BOMP). We
then derive thresholds on the block-sparsity level as a function of $\mub$ and $\nu$ for
both methods to recover the correct block-sparse $\bx$.
For L-OPT this complements the results in \cite{EM082}
that establish the recovery capabilities of L-OPT under the
condition that $\bbd$ satisfies a block-RIP
with a small enough restricted isometry constant. For the
special case of the columns of $\bbd[\ell]$ being orthonormal for
each $\ell$, we suggest a block-version of the MP 
algorithm \cite{Mallat93}, termed block-MP (BMP).

\subsection{Block OMP and block MP}

The BOMP algorithm begins by initializing the residual as $\brv_0=\by$. At the $\ell$th stage ($\ell \geq 1$)
we choose the block that is best matched to
$\brv_{\ell-1}$ according to:
\begin{equation}
\label{eq:ix} i_{\ell}=\argmax_{i} \|\bbd^H[i] \brv_{\ell-1}\|_2.
\end{equation}
Once the index $i_\ell$ is chosen, we find $\bx_{\ell}[i]$ as the solution to
\begin{equation}
\label{eq:lsq} \min \left\|\by-\sum_{i \in \I}\bbd[i]\bx_{\ell}[i]\right\|_{2}
\end{equation}
where $\I$ is the set of chosen indices $i_j,1 \leq j \leq \ell$. The
residual is then updated as
\begin{equation}
\label{eq:rlsq} \brv_{\ell}=\by-\sum_{i \in \I}\bbd[i]\bx_{\ell}[i].
\end{equation}

In the special case of the columns of $\bbd[\ell]$ being
orthonormal for each $\ell$ (the elements across different blocks
do not have to be orthonormal), we consider an extension of the MP
algorithm to the block-case. The resulting
algorithm, termed BMP, starts by initializing the residual as
$\brv_0=\by$ and at the $\ell$th stage ($\ell \geq 1$) chooses the
block that is best matched to $\brv_{\ell-1}$ according to
\eqref{eq:ix}. Then, however, the algorithm does not perform a
least-squares minimization over the blocks that have already been
selected, but directly updates the residual according to
\begin{align}
        \brv_{\ell}=\brv_{\ell-1}-\bbd[i_\ell]\bbd^H[i_\ell]\brv_{\ell-1}.
\end{align}

\subsection{Recovery conditions}

Our main result, summarized in Theorems~\ref{thm:sc} and
\ref{thm:mu} below, is that any block $k$-sparse vector $\bx$ can
be recovered from measurements $\by=\bbd\bx$ using either the BOMP
algorithm or L-OPT if the block-coherence  satisfies
$kd<(\mub^{-1}+d-(d-1)\nu\mub^{-1})/2$. In the special case of the
columns of $\bbD[\ell]$ being orthonormal for each $\ell$, we have
$\nu=0$ and therefore the recovery condition becomes
$kd<(\mub^{-1}+d)/2$. In this case BMP exhibits exponential convergence rate (see Theorem \ref{thm:BMP}).
If the block-sparse vector $\bx$ was treated as a
(conventional) $kd$-sparse vector without exploiting knowledge of
the block-sparsity structure, a sufficient condition for perfect
recovery using OMP \cite{Tropp04} or (\ref{eq:l1}) for $d=1$
(known as BP) is $kd<(\mu^{-1}+1)/2$. Comparing with
$kd\,<\,(\mub^{-1}+d)/2$, we can see that, thanks to
$\mub\,\le\,\mu$, making explicit use of block-sparsity leads to
guaranteed recovery for a potentially higher sparsity level.
Later, we will establish conditions for such a result to hold even
when $\nu\,\neq\,0$.

To formally state our main results, suppose that $\bx_0$ is a length-$N$
block $k$-sparse vector, and let $\by=\bbd\bx_0$.
 Let $\bbd_0$ denote the $L \times (kd)$ matrix whose blocks
correspond to the nonzero blocks of $\bx_0$, and let $\obbd_0$ be
the matrix of size $L \times (N-kd)$ which contains the $L\,\times\,d$ blocks of
$\bbd$ that are not in $\bbd_0$. We then have the following theorem proved in Section~\ref{sec:proofth2}.
\begin{theorem}
\label{thm:sc} Let $\bx_0\,\in\,\CC^{N}$ be a block $k$-sparse vector
with blocks of length $d$, and let $\by=\bbd\bx_0$ for a given $L
\times N$ matrix $\bbd$. A sufficient condition for the BOMP and the
L-OPT algorithm to recover $\bx_0$ is that
\begin{equation}
\label{eq:sc} \rho_c(\bbd_0^\dagger \obbd_0) <1
\end{equation}
where
\begin{equation}\label{eq:rhocdefi}
\rho_c(\bba)=\max_{r} \sum_\ell \rho(\bba[\ell,r])
\end{equation}
and $\bba[\ell,r]$ is the $(\ell,r)$th $d\,\times\,d$ block of
$\bba$. In this case, BOMP picks up a correct new block in each step,
and consequently converges in at most $k$ steps.
\end{theorem}
Note that
\begin{equation}
\rho_c( \bbd_0^\dagger \obbd_0)=\max_{r} \rho_c(\bbd_0^\dagger
\obbd_0[r]).
\end{equation}
Therefore, (\ref{eq:sc}) implies that for all $r$,
\begin{equation}
\label{eq:scs} \rho_c(\bbd_0^\dagger \obbd_0[r])<1.
\end{equation}

The sufficient condition (\ref{eq:sc}) depends on $\bbd_{0}$
and therefore on the location of the nonzero blocks in $\bx_{0}$,
which, of course, is not known in advance. Nonetheless, as the
following theorem, proved in Section~\ref{sec:proofth2}, shows,
(\ref{eq:sc}) holds universally under certain conditions on $\mub$ and $\nu$ associated with
the dictionary $\bbd$.
\begin{theorem}
\label{thm:mu} Let $\mub$ be the block-coherence and $\nu$ the sub-coherence of the dictionary $\bbd$. Then (\ref{eq:sc}) is satisfied if
\begin{equation}
\label{eq:muc1} kd<\frac{1}{2} \left(\mub^{-1}+d-(d-1)\frac{\nu}{\mub}\right).
\end{equation}
\end{theorem}
\noindent For $d=1$, and therefore $\nu=0$, we recover the
corresponding condition $k\,<\,(\mu^{-1}+1)/2$ reported in
\cite{Tropp04,DE03}. In the special case where the columns of
$\bbd[\ell]$ are orthonormal for each $\ell$, we have $\nu=0$ and
\eqref{eq:muc1} becomes
\begin{align}
\label{eq:orthc}
        kd<\frac{1}{2}(\mub^{-1}+d).
\end{align}

The next theorem shows that under condition (\ref{eq:orthc}),
BMP exhibits exponential convergence rate in the case where each block $\bbd[\ell]$ consists of
orthonormal columns.
\begin{theorem}\label{thm:BMP}
If $\bbd^H[\ell]\bbd[\ell]=\bbi_d$, for all $\ell$, and
$kd\,<\,(\mub^{-1}+d)/2$, then we have:
\begin{enumerate}
\item BMP picks up a correct block in each step.
\item The energy of the residual decays exponentially, i.e., $\|\brv_{\ell}\|_2^2\leq \beta^\ell \|\brv_{0}\|_2^2$ with
\begin{align}
        \beta = 1-\frac{1-(k-1)d\mub}{k}.
\end{align}
\end{enumerate}
\end{theorem}

\section{Proofs of Theorems~\ref{thm:sc}, \ref{thm:mu}, and \ref{thm:BMP}} \label{sec:proofth2}


Before proceeding with the actual proofs, we start with
some definitions and basic results that will be used throughout this section.

For $\bx \, \in \,\CC^N$, we define the general mixed
$\ell_2/\ell_p$-norm ($p=1,2,\infty$ here and in the following):
\begin{equation}
\label{eq:2p} \|\bx\|_{2,p} = \|\bv\|_p, \quad \mbox{where }
v_{\ell}=\|\bx[\ell]\|_2
\end{equation}
and the $\bx[\ell]$ are consecutive length-$d$ blocks.
For an $L \times N$ matrix $\bba$ with $L=Rd$
and $N=Md$, where $R$ and $M$ are integers, we define the mixed
matrix norm (with block size $d$) as
\begin{equation}
\label{eq:2infm} \|\bba\|_{2,p} = \max_{\bx\,\neq\,{\bf 0}}
\frac{\|\bba\bx\|_{2,p}}{\|\bx\|_{2,p}}.
\end{equation}

The following lemma provides bounds on $\|\bba\|_{2,p}$,
which will be used in the sequel.
\begin{lemma}
\label{lemma:norms} Let $\bba$ be an $L \times N$ matrix with
$L=Rd$ and $N=Md$. Denote by $\bba[\ell,r]$ the $(\ell, r)$th $d\,\times\, d$ block of
$\bba$.
Then,
\vspace*{-2mm}
\begin{eqnarray}
\label{eq:lemi} \|\bba\|_{2,\infty} & \leq  & \max_\ell
\sum_{r} \rho(\bba[\ell,r])\, \deft \, \rho_r(\bba)  \\
\label{eq:lem1} \|\bba\|_{2,1} & \leq & \max_{r} \sum_\ell
\rho(\bba[\ell,r])\, \deft \, \rho_c(\bba).
\end{eqnarray}
\vspace*{-2mm}
In particular, $\rho_r(\bba)=\rho_c(\bba^H)$.
\end{lemma}
\begin{IEEEproof}
See Appendix~\ref{app:mnorm}.
\end{IEEEproof}

\begin{lemma}
\label{lemma:mnorm} $\rho_{c}(\bba)$  as defined in \eqref{eq:rhocdefi} is a matrix norm and as such satisfies the following properties:
\begin{itemize}
\item
Nonnegative: $\rho_{c}(\bba)\,\ge\,0$
\item
Positive: $\rho_{c}(\bba)=0$ if and only if $\bba={\bf 0}$
\item
Homogeneous: $\rho_{c}(\alpha \bba)=|\alpha|\rho_{c}(\bba)$ for all $\alpha\,\in\,\CC$
\item
Triangle inequality: $\rho_{c}(\bba+\bbb)\,\le\,\rho_{c}(\bba)+\rho_{c}(\bbb)$
\item
Submultiplicative: $\rho_{c}(\bba \bbb)\,\le\,\rho_{c}(\bba)\rho_{c}(\bbb)$.
\end{itemize}
\end{lemma}

\begin{IEEEproof}
See Appendix~\ref{app:matrixnorm}.
\end{IEEEproof}

\subsection{Proof of Theorem~\ref{thm:sc} for BOMP}

We begin by proving that (\ref{eq:sc}) is sufficient to ensure
recovery using the BOMP algorithm. We first show that if $\brv_{\ell-1}$ is in $\R(\bbd_0$),
then the next chosen index
$i_{\ell}$ will correspond to a block
in $\bbd_0$. Assuming that this is true, it follows immediately that
$i_{1}$ is correct since clearly $\brv_0=\by$ lies in
$\R(\bbd_0)$.  Noting that $\brv_{\ell}$ lies in the space spanned
by $\by$ and $\bbd[i],i \in \I_{\ell}$, where $\I_{\ell}$
denotes the indices chosen up to stage $\ell$, it follows that if
$\I_{\ell}$ corresponds to correct indices, i.e., $\bbd[i]$ is a
block of $\bbd_0$ for all $i \in \I_{\ell}$, then $\brv_{\ell}$
also lies in $\R(\bbd_0$) and the next index will be correct as well.
Thus, at every step a correct $L\,\times\,d$ block of $\bbd$ is selected. As we will
show below no index will be chosen twice since the new residual is
orthogonal to all the previously chosen subspaces; consequently
the correct $\bx_0$ will be recovered in $k$ steps.

We first show that if $\brv_{\ell-1} \in \R(\bbd_0)$,
then under (\ref{eq:sc}) the next chosen index corresponds to a block
in $\bbd_0$. This is equivalent to requiring that
\begin{equation}
\label{eq:condp} z(\brv_{\ell-1})=\frac{\|\obbd_0^H
\brv_{\ell-1}\|_{2,\infty}}{\|\bbd_0^H \brv_{\ell-1}\|_{2,\infty}}<1.
\end{equation}
From the properties of the pseudo-inverse, it follows that $\bbd_{0}\bbd^{\dagger}_{0}$ is the orthogonal projector onto $\R(\bbd_0)$. Hence, it holds that
$\bbd_0 \bbd_0^\dagger \brv_{\ell-1}=\brv_{\ell-1}$. Since $\bbd_0\bbd_0^\dagger$ is
Hermitian,
we have
\begin{equation}
\label{eq:dagger} (\bbd_0^\dagger)^H \bbd_0^H
\brv_{\ell-1}=\brv_{\ell-1}.
\end{equation}
Substituting (\ref{eq:dagger}) into (\ref{eq:condp}) yields
\begin{align}
z(\brv_{\ell-1})&=\frac{\|\obbd_0^H (\bbd_0^\dagger)^H \bbd_0^H
\brv_{\ell-1}\|_{2,\infty}}{\|\bbd_0^H
\brv_{\ell-1}\|_{2,\infty}} \nonumber \\
&\leq \rho_r(\obbd_0^H
(\bbd_0^\dagger)^H) \nonumber \\
&=\rho_c( \bbd_0^\dagger\obbd_0)
\end{align}
where we used Lemma~\ref{lemma:norms}. 

It remains to show that BOMP in each step 
chooses a new block participating in the (unique) representation $\by=\bbd \bx$. 
We start by defining $\bbd_\ell=[\bbd[i_1]\,\cdots\,\bbd[i_\ell]]$ where $i_j\,\in\,{\cal I},\, 1\,\le\,j\,\le\,\ell$. 
It follows that the solution of the minimization problem in (\ref{eq:lsq}) is given by
$$
\hat{\bx}=(\bbd^H_\ell \bbd_\ell)^{-1}\bbd^H_\ell \by
$$
which upon inserting into (\ref{eq:rlsq}) yields
$$
\brv_\ell=(\bbI-\bbd_\ell (\bbd^H_\ell \bbd_\ell)^{-1}\bbd^H_\ell)\by.
$$
Now, we note that $\bbd_\ell (\bbd^H_\ell \bbd_\ell)^{-1}\bbd^H_\ell$ is the orthogonal projector
onto the range space of $\bbd_\ell$. Therefore $\|\bbd^H[i]\brv_\ell\|_2=0$ for all blocks
$\bbd[i]$ that lie in the span of the matrix $\bbd_\ell$. By the assumption
in Proposition (\ref{prop:inv}) we are guaranteed that as long as $l\,<\,k$ there exists at least one
block (in $\bbd_0$) which does not lie in the span of $\bbd_\ell$. Since this block (or these blocks) will lead to strictly
positive $\|\bbd^H[i]\brv_{\ell}\|_2$ the result is established. This concludes the proof.

\subsection{Proof of Theorem~\ref{thm:sc} for L-OPT}

We next show that (\ref{eq:sc}) is also sufficient to ensure recovery
using L-OPT.
To this end we rely on the
following lemma:
\begin{lemma}
\label{lemma:1inq} Suppose that $\bv\,\in\,\CC^{kd}$ 
with $\|\bv[\ell]\|_2>0$, for all $\ell$, and that $\bba$ is a matrix
of size $L \times (kd)$, with $L=Rd$ and the $d\,\times\,d$ blocks $\bba[\ell,r]$.
Then, $\|\bba\bv\|_{2,1} \leq
\rho_c(\bba) \|\bv\|_{2,1}$. If in addition the values of
$\rho_c(\bba\bbj_{\ell})$ are not all equal, then the inequality
is strict. Here, $\bbj_{\ell}$ is a $(kd) \times d$ matrix that is
all zero except for the $\ell$th $d\,\times\,d$ block which equals
$\bbi_d$.
\end{lemma}
\begin{IEEEproof}
See Appendix~\ref{app:1inq}.
\end{IEEEproof}

To prove that L-OPT recovers the correct vector $\bx_0$, let
$\bx'\,\neq\,\bx_0$ be another length-$N$ block $k$-sparse vector
for which $\by=\bbd\bx'$.
Denote by $\bc_0$ and $\bc'$ the length-$kd$ vectors consisting of
the nonzero elements of $\bx_0$ and $\bx'$, respectively. Let
$\bbd_0$ and $\bbd'$ denote the corresponding columns of $\bbd$
so that $\by=\bbd_0\bc_0=\bbd'\bc'$. From the assumption in
Proposition~\ref{prop:inv}, it follows that there cannot be two
different representations using the same blocks $\bbd_0$.
Therefore, $\bbd'$ must contain at least one block, $\bbz$, that
is not included in $\bbd_0$. From (\ref{eq:scs}), we get
$\rho_c(\bbd_0^\dagger \bbz)<1$. For any other block $\bbu$ in
$\bbd$, we must have that
\begin{equation}
\label{eq:nl1} \rho_c(\bbd_0^\dagger \bbu)\leq 1.
\end{equation}
Indeed, if $\bbu \in \bbd_0$, then
$\bbu=\bbd_0[\ell]=\bbd_0\bbj_{\ell}$ where $\bbj_{\ell}$ was
defined in Lemma~\ref{lemma:1inq}. In this case, $\bbd_0^\dagger
\bbd_0[\ell]=\bbj_{\ell}$ and therefore
$\rho_c(\bbd_{0}^{\dagger}\bbu)=\rho_c(\bbj_\ell)=1$. If, on the
other hand, $\bbu=\obbd[\ell]$ for some $\ell$, then it follows
from (\ref{eq:scs}) that $\rho_c(\bbd_0^\dagger \bbu)<1$.

Now, suppose first that the $(kd)\,\times\,d$ blocks in
$\bbd_0^\dagger \bbd'$ do not all have the same\footnote{Note that for an $(sd)\,\times\,d$ matrix $\bba$,
$\rho_{c}(\bba)=\sum_{\ell}\rho(\bba[\ell])$, where
$\bba[\ell],\,\ell=1,2,...,s$, denotes the $d\,\times\,d$ block of
$\bba$ made up of the rows $\{(\ell-1)d+1,...,\ell d\}.$}
$\rho_{c}$. Then,
\begin{align}
\label{eq:l1u} \|\bc_0\|_{2,1}  & =   \|\bbd_0^\dagger \bbd_0 \bc_0\|_{2,1}\\
&= \|\bbd_0^\dagger
\bbd'\bc'\|_{2,1} \nonumber \\
& <  \rho_c(\bbd_0^\dagger \bbd')\|\bc'\|_{2,1} \label{eq:l1u_ineq}\\
& \leq
\|\bc'\|_{2,1} \label{eq:l1u_second}
\end{align}
where the first equality is a consequence of the columns of
$\bbd_0$ being linearly independent (a consequence of
the assumption in Proposition~\ref{prop:inv}), the first inequality follows from
Lemma~\ref{lemma:1inq} since $\|\bc'[\ell]\|_2>0$,  for all $\ell$, and the last
inequality follows from (\ref{eq:nl1}). If all the $(kd)\,\times\,d$ blocks in
$\bbd_0^\dagger \bbd'$ have identical $\rho_{c}$, then the inequality (\ref{eq:l1u_ineq}) is no
longer strict, but the second
inequality (\ref{eq:l1u_second}) becomes strict instead as a consequence of $\rho_c(\bbd_0^\dagger \bbz)\,<\,1$; therefore $\|\bc_{0}\|_{2,1}\,<\,\|\bc^{\prime}\|_{2,1}$ still holds.

Since $\|\bx_0\|_{2,1}=\|\bc_0\|_{2,1}$ and
$\|\bx'\|_{2,1}=\|\bc'\|_{2,1}$, we conclude that under
(\ref{eq:scs}), any set of coefficients used to represent the
original signal that is not equal to $\bx_0$ will result in a
larger $\ell_2/\ell_1$-norm.

\subsection{Proof of Theorem~\ref{thm:mu}}

We start by deriving an upper bound on
$\rho_c(\bbd_0^\dagger \obbd)$ in terms of $\mub$ and $\nu$. Writing $\bbd_0^\dagger$ out,
we have that
\begin{equation}
\label{eq:bound1} \rho_c(\bbd_0^\dagger
\obbd)=\rho_c((\bbd_0^H\bbd_0)^{-1}\bbd_0^H \obbd).
\end{equation}
Submultiplicativity of $\rho_{c}(\bba)$ (Lemma \ref{lemma:mnorm}) implies that
\begin{align}
\label{eq:bound2} \rho_c(\bbd_0^\dagger \obbd) & \leq
\rho_c((\bbd_0^H\bbd_0)^{-1})\rho_c(\bbd_0^H \obbd) \nonumber \\
& =  \rho_c((\bbd_0^H\bbd_0)^{-1}) \max_{j \notin \Lambda_0}
\sum_{i \in \Lambda_0} \rho(\bbd^H[i] \bbd[j])
\end{align}
where $\Lambda_0$ is the set of indices $\ell$ for which
$\bbd[\ell]$ is in $\bbd_0$. Since $\Lambda_0$ contains $k$
indices, the last term in (\ref{eq:bound2}) is bounded above by
$kd\mub$, which allows us to conclude that
\begin{equation}
\label{eq:bound3} \rho_c(\bbd_0^\dagger \obbd) \leq
\rho_c((\bbd_0^H\bbd_0)^{-1}) kd\mub.
\end{equation}

It remains to develop a bound on $\rho_c((\bbd_0^H\bbd_0)^{-1})$.
To this end, we express $\bbd_0^H\bbd_0$ as
$\bbd_0^H\bbd_0=\bbi+\bba$, where $\bba$ is a $(kd) \times (kd)$
matrix with blocks $\bba[\ell,r]$ of size $d \times d$ such that $\bba_{i,i}=0$,
 for all $i$. This follows from the fact
that the columns of $\bba$ are normalized.
Since
$\bba[\ell,r]=\bbd_0^H[\ell]\bbd_0[r]$, for all $\ell\neq r$, and $\bba[r,r]=\bbd_0^H[r]\bbd_0[r]-\bbi_d$,
we have
\begin{align}
\rho_c(\bba)&=\max_{r} \sum_{\ell} \rho(\bba[\ell,r]) \nonumber\\
&\leq \max_r  \rho(\bba[r,r]) + \max_r\sum_{\ell \neq r} \rho(\bba[\ell,r])\label{eq:maxrho}\\
&\leq (d-1)\nu +(k-1)d\mub\label{eq:Gershgorin2}
\end{align}
where the first term in \eqref{eq:Gershgorin2} is obtained by applying Ger\v{s}gorin's disc theorem (\cite[Corollary 6.1.5]{Horn85}) together with the definition of $\nu$; 
the second
term in \eqref{eq:Gershgorin2} follows from the fact that the summation in the second term of \eqref{eq:maxrho} is over $k-1$ elements and $\rho(\bba[\ell,r])$, for all $\ell\neq r$, can be upper-bounded by $d\mub$.
Assumption (\ref{eq:muc1}) now implies that
$(d-1)\nu+(k-1)d\mub<1$ and therefore, from
\eqref{eq:Gershgorin2}, we have $\rho_c(\bba)<1$.

We next use the following result.

\begin{lemma}
\label{lemma:neumann}
Suppose that $\rho_c(\bba)<1$. Then
$(\bbi+\bba)^{-1}=\sum_{k=0}^\infty (-\bba)^k$.
\end{lemma}

\begin{IEEEproof}
Follows immediately by using the fact that $\rho_{c}(\bba)$ is a matrix norm (cf.~Lemma~\ref{lemma:mnorm}) and applying \cite[Corollary 5.6.16]{Horn85}.
\end{IEEEproof}

Thanks to Lemma~\ref{lemma:neumann}, we have that
\begin{align}
        \rho_c((\bbd_0^H\bbd_0)^{-1}) & = \rho_c\left(\sum_{k=0}^\infty
(-\bba)^k\right) \nonumber \\
& \leq \sum_{k=0}^\infty \left(\rho_c(\bba)\right)^k\label{eq:subandtriangle}\\
& =\frac{1}{1-\rho_c(\bba)}\nonumber\\
&\leq \frac{1}{1-(d-1)\nu-(k-1)d\mub}.\label{eq:boundonrhoc}
\end{align}
Here, \eqref{eq:subandtriangle} is a consequence of $\rho_{c}(\bba)$ satisfying the triangle inequality and being submultiplicative and \eqref{eq:boundonrhoc} follows by using \eqref{eq:Gershgorin2}.


Combining \eqref{eq:boundonrhoc} with \eqref{eq:bound3}, we get
\begin{equation}
\rho_c(\bbd_0^\dagger \obbd) \leq \frac{ kd\mub}{1-(d-1)\nu-(k-1)d\mub} <1
\end{equation}
where the last inequality is a consequence of (\ref{eq:muc1}).

\subsection{Proof of Theorem~\ref{thm:BMP}}

The proof of the first part of Theorem \ref{thm:BMP} follows from the arguments in the proofs of Theorems \ref{thm:sc} and \ref{thm:mu} for $\nu=0$. As a
consequence of the first statement of Theorem \ref{thm:BMP}, we get that the residual $\brv_{\ell}$ in each step of the algorithm will be in $\R(\bbd_0)$. For the proof of the second statement in Theorem \ref{thm:BMP}, we mimic
the corresponding proof in \cite{gribonval2006}. We first need the following lemma, which is an extension of \cite[Lemma 3.5]{DeVore1996} to the block-sparse case. This lemma will provide us
with a lower bound on the amount of energy that can be removed from the residual $\brv_{\ell}$ in one step of the BMP algorithm.
\begin{lemma}
\label{lem:deVore} Let $\bbd_0$ denote the $L \times (kd)$ matrix whose blocks
correspond to the nonzero blocks of $\bx_0$. Then, we have
\begin{align}\label{eq:lemmadeVore}
        \max_{i}\|\bbd_0^H[i]\brv_{\ell}\|_2\geq\frac{\|\brv_{\ell}\|_{2}^2}{\|\bc_\ell\|_{2,1}}
\end{align}
where  $\bc_\ell$ is the coefficient vector corresponding to $\brv_{\ell}\neq \mathrm{\bf 0}$, i.e., $\brv_{\ell}=\bbd_0\bc_\ell$.
\end{lemma}
\begin{proof}
We start by noting that
$\brv_{\ell}=\bbd_{0}\bc_{\ell}=\sum_{i=1}^k\bbd_0[i]\bc_\ell[i]$,
where $\bc_{\ell}[i]\,\neq\,\mathrm{\bf 0}$ for at least one index
$i\,\in\,\{1,2,...,k\}$. It follows that
        \begin{align}
                \|\brv_{\ell}\|_{2}^2
                        & = \sum_{i=1}^k\bc_\ell^H[i]\bbd_0^H[i]\brv_{\ell} \nonumber\\
                        &\leq \sum_{i=1}^k|\bc_\ell^H[i]\bbd_0^H[i]\brv_{\ell}| \nonumber\\
                        &\leq \sum_{i=1}^k\|\bc_\ell[i]\|_2\|\bbd_0^H[i]\brv_{\ell}\|_2 \nonumber\\&  \leq \left(\max_{i}\|\bbd_0^H[i]\brv_{\ell}\|_2\right) \sum_{i=1}^k\|\bc_\ell[i]\|_2.
                         \label{eq:proofdo}
        \end{align}
The result then follows by noting that $
\sum_{i=1}^k\|\bc_\ell[i]\|_2=\|\bc_\ell\|_{2,1}$.
\end{proof}

Next, we compute an upper bound on $\|\bc_\ell\|_{2,1}$. Using
$\bbm[i,j]=\bbd_0^H[i]\bbd_0[j]$, where $i,j\in\{1,\ldots,k\}$, we
get
\begin{align}
        \|\brv_{\ell}\|_{2}^2 & = \bc_\ell^H\bbd_0^H\bbd_0\bc_\ell \nonumber \\
        &=\sum_{i=1}^k\sum_{j=1}^k\bc_\ell^H[i]\bbm[i,j]\bc_\ell[j] \nonumber \\
        & = \sum_{i=1}^k \bc_\ell^H[i]\bbm[i,i]\bc_\ell[i] + \sum_{i=1}^k\sum_{\substack{j=1\\j\neq i}}^k\bc_\ell^H[i]\bbm[i,j]\bc_\ell[j] \nonumber \\
        &\geq \sum_{i=1}^k\|\bc_\ell[i]\|_2^2-\sum_{i=1}^k\sum_{\substack{j=1\\j\neq i}}^k|\bc_\ell^H[i]\bbm[i,j]\bc_\ell[j]    |\label{eq:proof1}
\end{align}
where we used the fact that $\bbm[i,i]=\bbi_d$, for all $i$, as a consequence of each of the blocks of $\bbd_0$ consisting of orthonormal vectors.
Applying the Cauchy-Schwarz inequality to the second term in \eqref{eq:proof1}, we get
\begin{align}
        \|\brv_{\ell}\|_{2}^2&\geq \sum_{i=1}^k\|\bc_\ell[i]\|_2^2-\sum_{i=1}^k\sum_{\substack{j=1\\j\neq i}}^k\|\bc_\ell[i]\|_2\|\bbm[i,j]\bc_\ell[j]\|_2\label{eq:proof2}\\
        &\geq \|\bc_\ell\|_{2,2}^2-\sum_{i=1}^k\sum_{\substack{j=1\\j\neq i}}^k\|\bc_\ell[i]\|_2\|\bc_\ell[j]\|_2d\mub \label{eq:prooflow}\\
        & = \|\bc_\ell\|_{2,2}^2-d\mub\sum_{s=1}^{k-1}\sum_{i=1}^k\|\bc_\ell[i]\|_2\|\bc_\ell[(i+s)_k]\|_2\label{eq:proofmodulo}
\end{align}
where $(i+s)_k$ stands for $(i+s)$ modulo $k$, \eqref{eq:prooflow} follows from $\|\bbm[i,j]\bc_\ell[j]\|_2\,\le\, d\mub \|\bc_\ell[j]\|_2$, and \eqref{eq:proofmodulo} is obtained by merely rearranging terms in the summation in \eqref{eq:prooflow}.
Applying the Cauchy-Schwarz inequality to the inner product $\sum_{i=1}^k\|\bc_\ell[i]\|_2\|\bc_\ell[(i+s)_k]\|_2$, we obtain
\begin{align}
        \|\brv_{\ell}\|_{2}^2&\geq \|\bc_\ell\|_{2,2}^2-d\mub\sum_{s=1}^{k-1}\|\bc_\ell\|_{2,2}^2\label{eq:proof3}\\
        & = (1-(k-1)d\mub)\|\bc_\ell\|_{2,2}^2 \nonumber \\
        &\geq\frac{(1-(k-1)d\mub)}{k}\|\bc_\ell\|_{2,1}^2\label{eq:proof4}
\end{align}
where \eqref{eq:proof4} follows by the same argument as used in \eqref{eq:norm3}. Thus, combining \eqref{eq:proofdo} with \eqref{eq:proof4}, we get
\begin{align}
        \max_{i}\|\bbd_0^H[i]\brv_{\ell}\|_2\geq\sqrt{\frac{(1-(k-1)d\mub)}{k}}\|\brv_{\ell}\|_{2}.
\end{align}
Since, by the first statement in Theorem \ref{thm:BMP}, BMP picks a block in $\bbd_0$ in each step, we can bound the energy of the residual in the $(\ell+1)$st step as
\begin{align}
        \|\brv_{\ell+1}\|_{2}^2 &= \|\brv_{\ell}\|_{2}^2 - \|\bbd^H[i_{\ell+1}]\brv_{\ell}\|_2^2 \label{eq:orthogonality}\\
        &= \|\brv_{\ell}\|_{2}^2 - \max_{i}\|\bbd_0^H[i]\brv_{\ell}\|_2^2 \nonumber\\
        &\leq \left(1-\frac{(1-(k-1)d\mub)}{k}\right)\|\brv_{\ell}\|_{2}^2
\end{align}
where in \eqref{eq:orthogonality} we used the fact that $\brv_{\ell+1}$ is orthogonal to $\bbd[i_{\ell+1}]\bbd^H[i_{\ell+1}]\brv_{\ell}$. This concludes the proof.

\section{Discussion}
\label{sec:discuss}

Theorem~\ref{thm:mu} indicates under which conditions exploiting
block-sparsity leads to higher recovery thresholds than treating
the block-sparse signal as a (conventionally) sparse signal. For
dictionaries $\bbd$ where the individual blocks $\bbd[\ell]$
consist of orthonormal columns, for each $\ell$, we have $\nu=0$
and hence, thanks to $\mub\,\le\,\mu$, recovery through exploiting
block-sparsity is guaranteed for a potentially higher sparsity
level. If the individual blocks $\bbd[\ell]$ are, however, not
orthonormal, we have $\nu\,>\,0$, and (\ref{eq:muc1}) shows that $\nu$
has to be small for block-sparse recovery to result in higher
recovery thresholds than sparse recovery. It is now natural to
consider the case where one starts with a general dictionary
$\bbd$ and orthogonalizes the individual blocks $\bbd[\ell]$ so
that $\nu=0$. The comparison that is meaningful here is between
the recovery threshold of the original dictionary $\bbd$ without
exploiting block-sparsity and the recovery threshold of the
orthogonalized dictionary taking block-sparsity into account. To
this end, we start by noting that the assumption in Proposition
\ref{prop:inv} implies that the columns of $\bbd[\ell]$ are
linearly independent, for each $\ell$. We can therefore write
$\bbd[\ell]=\bba[\ell]\bbw_{\ell}$ where $\bba[\ell]$ consists of
orthonormal columns that span $\R(\bbd[\ell])$ and $\bbw_{\ell}$
is invertible. The orthogonalized dictionary is given by the $L
\times N$ matrix $\bba$ with blocks $\bba[\ell]$. Since
$\bbd=\bba\bbw$ with the $N\times N$ block-diagonal matrix $\bbw$
with blocks $\bbw_{\ell}$, we conclude that $\bc=\bbw\bx$ is
block-sparse and
---thanks to the invertibility of the $\bbw_{\ell}$--- of the same
block-sparsity level as $\bx$, i.e., orthogonalization preserves
the block-sparsity level. It is easy to see that the definition of block-coherence in
(\ref{eq:bc}) is invariant to the choice of orthonormal basis
$\bba[\ell]$ for $\R(\bba[\ell])$. This is because any other basis
has the form $\bba[\ell] \bbu_{\ell}$ for some unitary matrix
$\bbu_{\ell}$, and from the properties of the spectral norm
\begin{equation}
\rho(\bbm[\ell,r])=\rho(\bbu^{H}_{\ell} \bbm[\ell,r] \bbu_{r})
\end{equation}
for any unitary matrices $\bbu_{\ell},\bbu_{r}$. Unfortunately, it
seems difficult to derive general results on the relation between
$\mu$ before and $\mub$ after orthogonalization. Nevertheless, we
can establish a minimum block size $d$ above which
orthogonalization followed by block-sparse recovery leads to a
guaranteed improvement in the recovery thresholds. We first note
that the coherence $\mu$ of a dictionary consisting of $N=Md$ elements
in a vector space of dimension $L=Rd$ can be lower-bounded as
\cite{strohmer2003}
\begin{align}
        \mu\geq
\sqrt{\frac{M-R}{R(Md-1)}}\stackrel{Md\,\gg\,1}{\approx}\sqrt{\frac{M-R}{RMd}}.
\end{align}
Using this lower bound together with Proposition~\ref{prop:orthoblocks} and the fact that after orthogonalization we have $\nu=0$, 
it can be shown that if $d>RM/(M-R)$, then the recovery threshold obtained from taking block-sparsity into
account in the orthogonalized dictionary is higher than the recovery threshold corresponding to conventional sparsity in the original dictionary. 
This is true irrespectively of the dictionary we start from
as long as the dictionary satisfies the conditions of Proposition \ref{prop:inv}.

Finally, we note that finding dictionaries that lead to significant improvements in the recovery thresholds when exploiting block-sparsity seems to be a difficult design problem. For example, partitioning the realizations
of i.i.d. Gaussian matrices into blocks will, in general, not lead to satisfactory results. Nevertheless, there do exist dictionaries where significant improvements are possible.
Consider, for example, the pair of
bases $\bphi=\bbi_L$ and $\bpsi=\bbf \otimes \bbu_d$ shown in Section \ref{sec:uca}
to achieve the lower bound in (\ref{eq:bco4}).
For the corresponding dictionary $\bbd=[\bphi \,\, \bpsi]$, we have $M=2R$, $\mub=1/(d\sqrt{R})$, with the recovery threshold, assuming that block-sparsity is exploited, given by $kd<d(\sqrt{R}+1)/2$. The coherence of the dictionary is
$\mu = \|\textrm{vec}(\bbu_d)\|_\infty/\sqrt{R}$. Fig. \ref{fig:simu1}, obtained by averaging over randomly chosen unitary matrices $\bbu_d$, shows that the recovery thresholds obtained by taking block-sparsity into
account can be significantly higher than those for conventional sparsity. In particular, for
$\bbu_d=\bbi_d$, we obtain the conventional recovery threshold as $k=kd<(\sqrt{R}+1)/2$, which allows us to conclude that exploiting block-sparsity can result in guaranteed recovery for a sparsity level that is $d$ times higher than what would
be obtained in the (conventional) sparse case.
\begin{figure}[t]
\begin{center}
  \includegraphics[width=\linewidth]{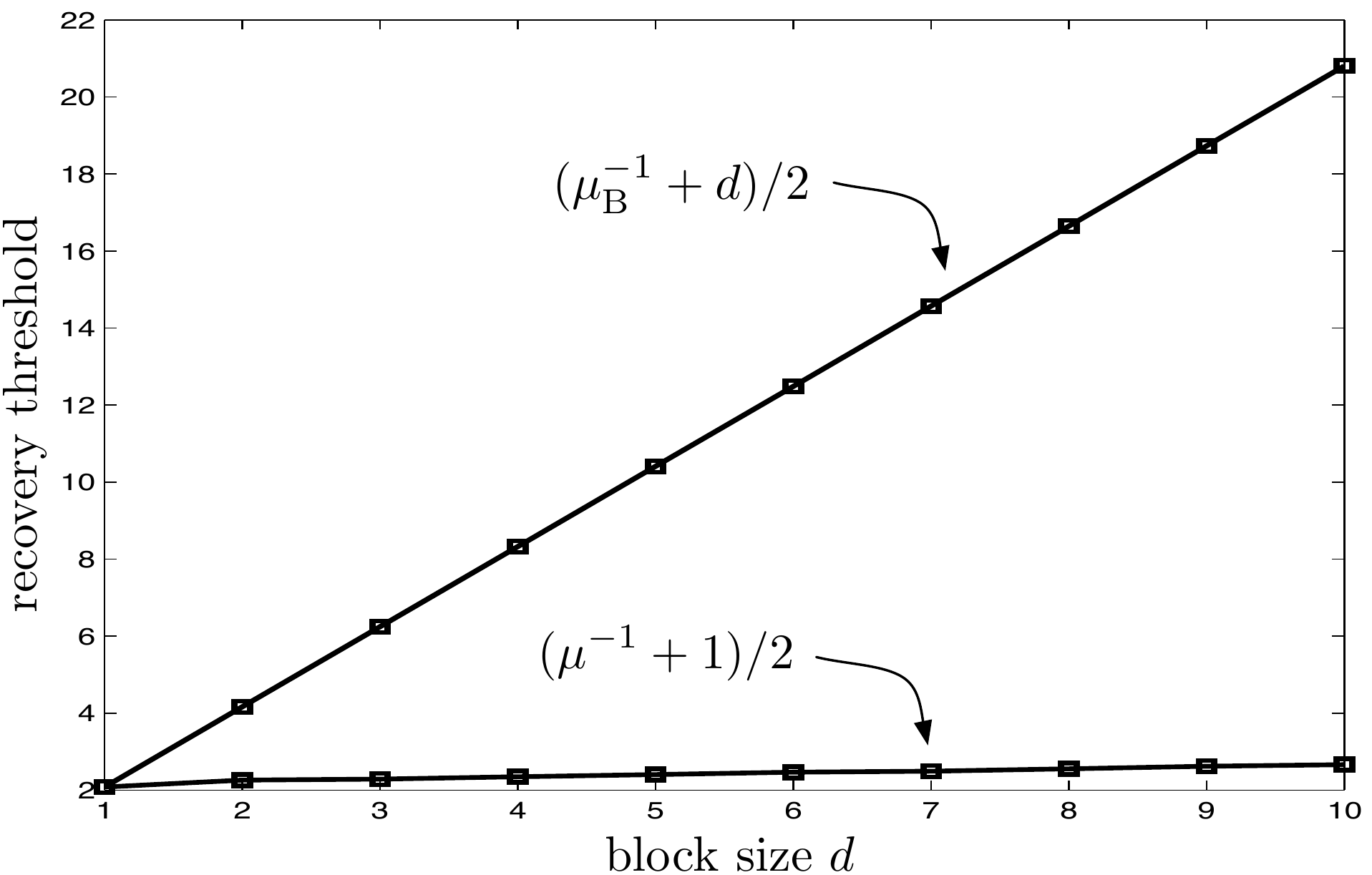}
  \caption{Recovery thresholds for both block-sparsity and conventional sparsity for $R=10$ as a function of $d$.
    }
  \label{fig:simu1}
  \end{center}
\end{figure}


\section{Numerical Results}
\label{sec:num}

The aim of this section is to quantify the improvement in the
recovery properties of OMP and BP obtained by taking
block-sparsity explicitly into account and performing recovery
using BOMP and L-OPT, respectively. 
In all simulation examples
below, we randomly generate dictionaries by drawing from i.i.d.
Gaussian matrices and normalizing the resulting columns to $1$.
The dictionary is divided into consecutive blocks of length $d$.
The sparse vector to be recovered has i.i.d. Gaussian entries on
the randomly chosen support set (according to a uniform prior).

In Figs. \ref{fig:BOMP10_100} and \ref{fig:BOMP10_20}, we plot the
recovery success rate\footnote{Success is declared if the
recovered vector is within a certain small Euclidean distance of
the original vector.} as a function of the block-sparsity level of
the signal to be recovered. For each block-sparsity level we
average over $1000$ pairs of realizations of the dictionary and
the block-sparse signal. We can see that BOMP outperforms OMP
significantly and BOMP with orthogonalized blocks, denoted as
BOMP-O, yields slightly better performance than BOMP. We also
evaluate the performance of L-OPT compared to BP, as well as L-OPT
run on orthogonalized blocks, termed L-OPT-O. For each
block-sparsity level we average over $200$ pairs of realizations
of the dictionary and the block-sparse signal. The corresponding
results, depicted in Figs. \ref{fig:LOPT10_100} and
\ref{fig:LOPT10_20}, show that L-OPT outperforms BP, and L-OPT-O
slightly outperforms L-OPT. Furthermore, we can see that BOMP-O
significantly outperforms L-OPT-O.

\begin{figure}[t]
\begin{center}
  \includegraphics[width=\linewidth]{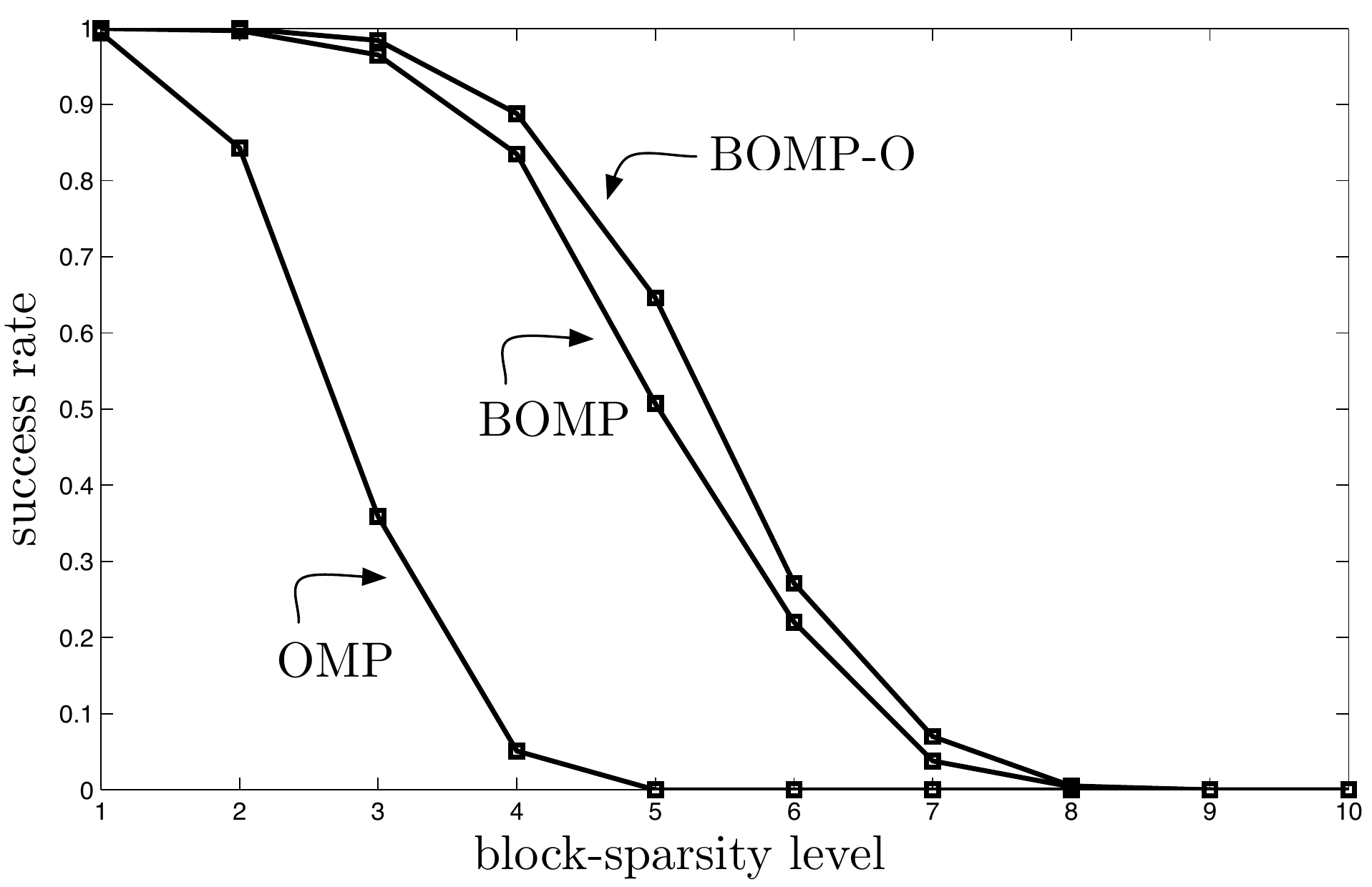}
  \caption{Performance of OMP, BOMP, and BOMP-O for a dictionary with $L=40,N=400$, and $d=4$.
    }
  \label{fig:BOMP10_100}
  \end{center}
\end{figure}

\begin{figure}[t]
\begin{center}
 \includegraphics[width=\linewidth]{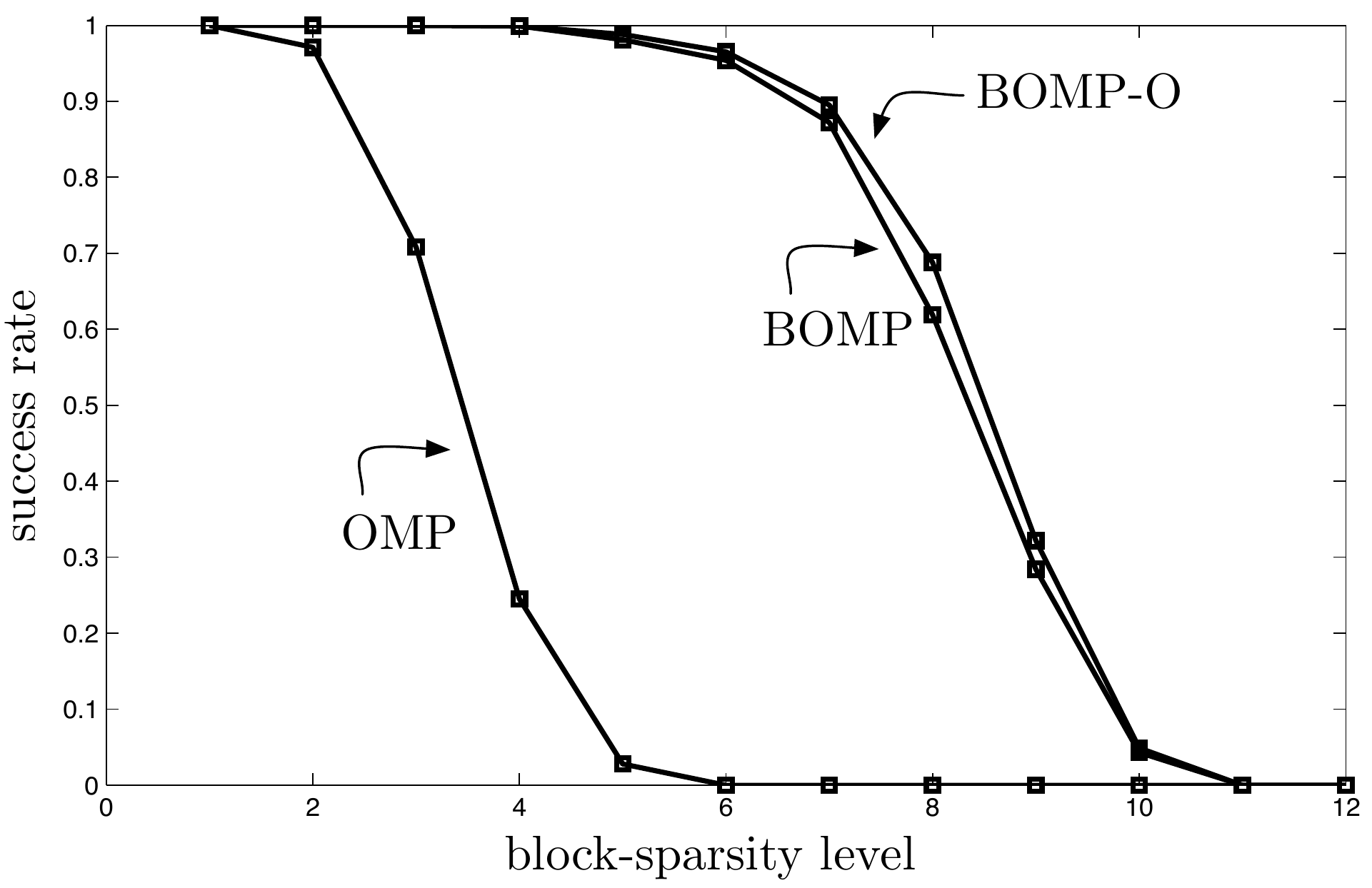}
  \caption{Performance of OMP, BOMP, and BOMP-O for a dictionary with $L=80,N=160$, and $d=8$.
    }
  \label{fig:BOMP10_20}
  \end{center}
\end{figure}

\begin{figure}[t]
\begin{center}
  \includegraphics[width=\linewidth]{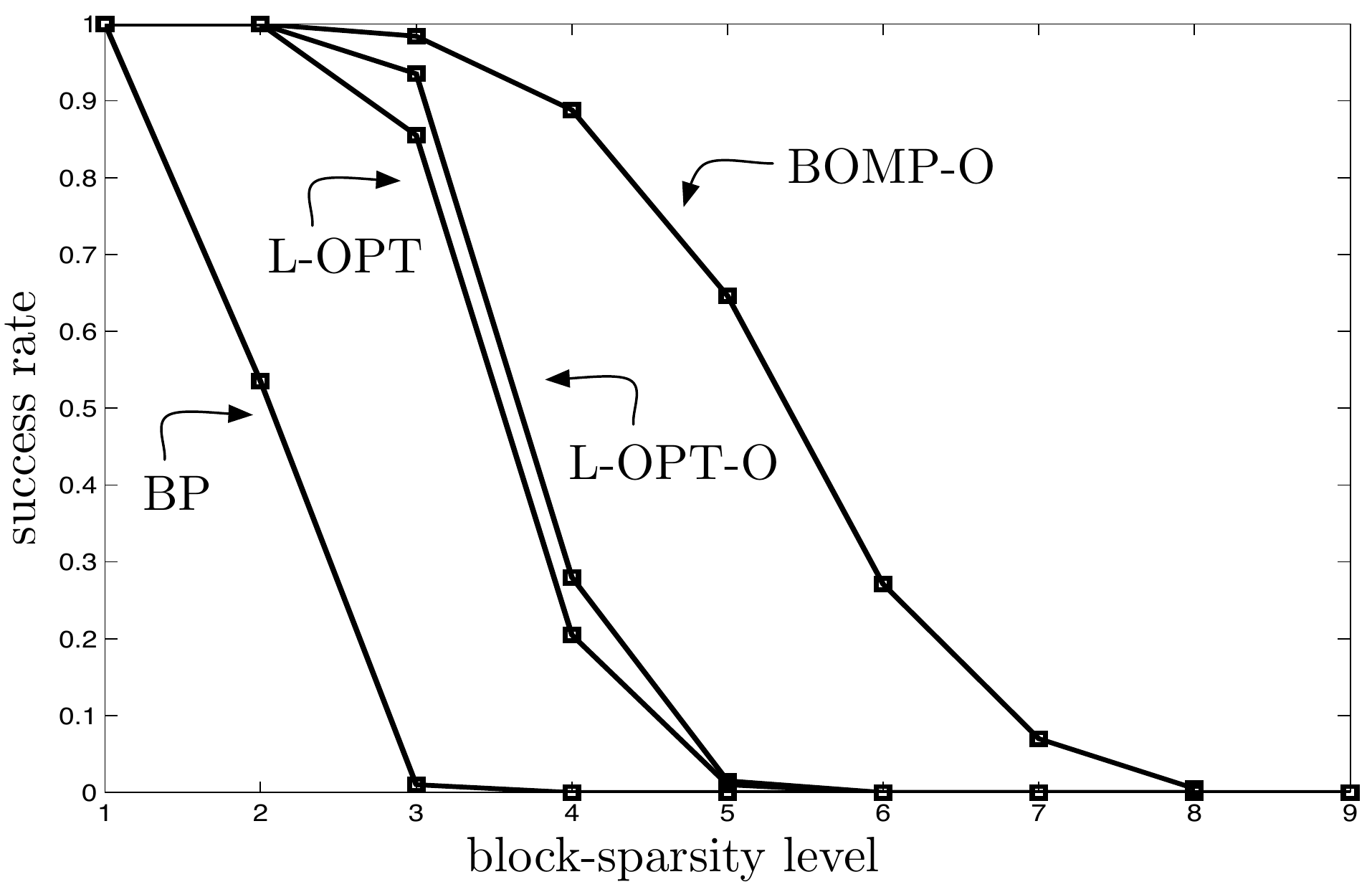}
  \caption{Performance of BP, L-OPT, L-OPT-O, and BOMP-O for a dictionary with $L=40,N=400$, and $d=4$.
    }
  \label{fig:LOPT10_100}
  \end{center}
\end{figure}


\begin{figure}[t]
\begin{center}
  \includegraphics[width=\linewidth]{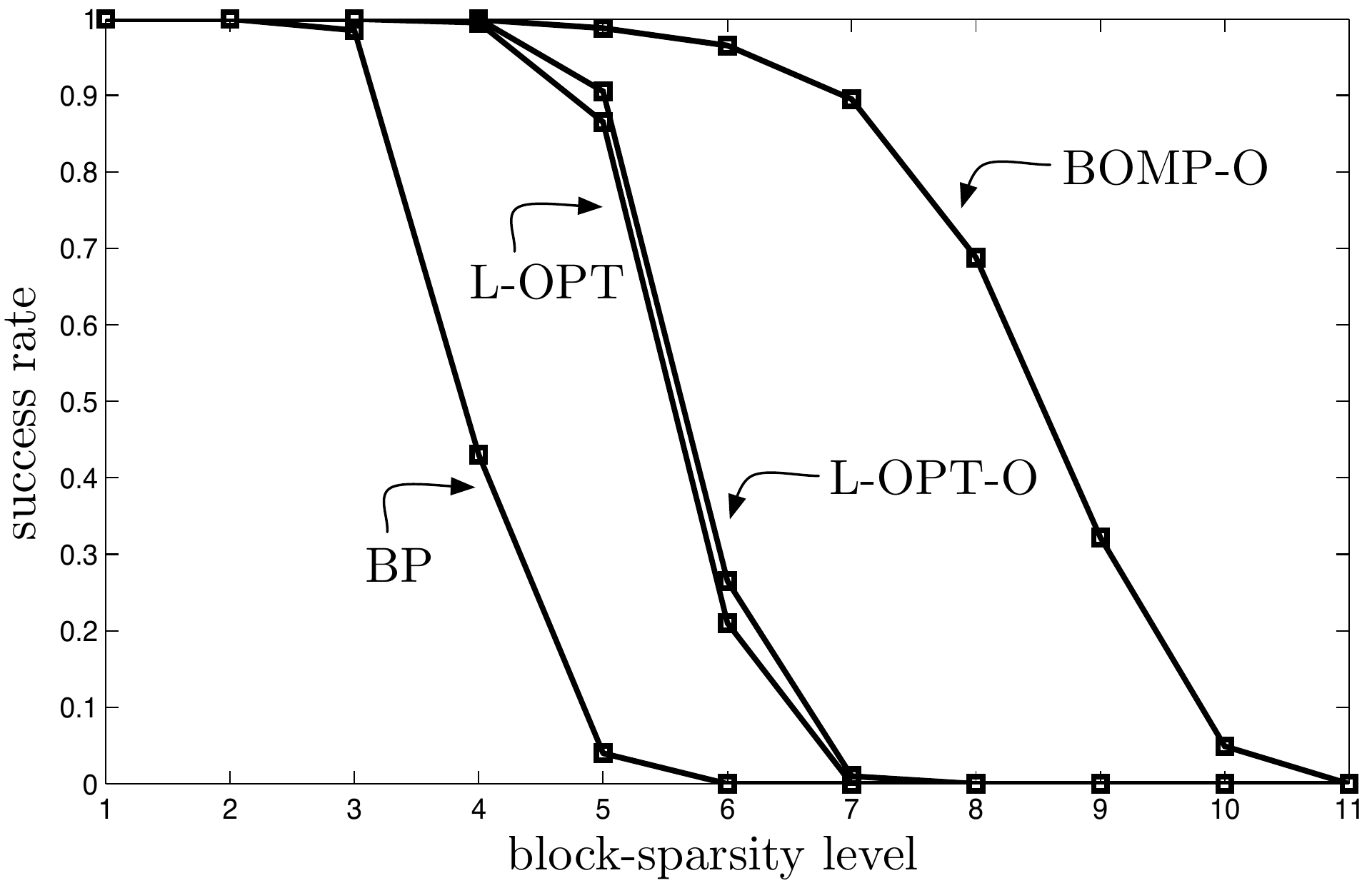}
  \caption{Performance of BP, L-OPT, L-OPT-O, and BOMP-O for a dictionary with $L=80,N=160$, and $d=8$.
    }
  \label{fig:LOPT10_20}
  \end{center}
\end{figure}

\section{Conclusion}
\label{sec:conclude}

This paper extends the concepts of uncertainty relations, coherence,
and recovery thresholds for matching pursuit and basis pursuit to the case of sparse
signals that have additional structure, namely block-sparsity. The extension is made possible by an appropriate definition of block-coherence.

The motivation for considering block-sparse signals is two-fold.
First, in many applications the nonzero elements of sparse vectors tend to cluster
in blocks; several examples are given in \cite{EM082}. Second, it
is shown in \cite{EM082} that sampling problems over
unions of subspaces can be converted into block-sparse recovery
problems. Specifically, this is true when the union has a
direct-sum decomposition, which is the case in many applications
including multiband signals
\cite{Eldar2009,MishaliEldar2009,ME09,Landau67}.
Reducing union of subspaces problems to block-sparse recovery
problems allows for the first general class of concrete recovery
methods for union of subspace problems. This was the main
contribution of \cite{EM082} together with equivalence and
robustness proofs for L-OPT based on a suitably modified
definition of the restricted isometry property. Here, we
complement this contribution by developing similar results using
the concept of block-coherence.





\appendices

\section{Proof of Lemma~\ref{lemma:norms}}
\label{app:mnorm}

We first prove (\ref{eq:lemi}): 
\begin{eqnarray}
\label{eq:lemi1} \|\bba\bx\|_{2,\infty} & = & \max_j \left\|
\sum_i \bba[j,i]\bx[i]\right\|_2  \nonumber \\ 
& \leq &   \max_j \sum_i \left\|\bba[j,i]\bx[i]\right\|_2 \nonumber \\
& \leq &
\max_j\sum_{i}\|\bx[i]\|_2\, \rho(\bba[j,i]) \nonumber \\
& \leq & \|\bx\|_{2,\infty} \max_j\sum_{i}\rho(\bba[j,i]).
\end{eqnarray}
Therefore, for any $\bx\,\in\,\CC^N$ with $\bx\,\neq\,{\bf 0}$, we have
\begin{equation}
\frac{\|\bba\bx\|_{2,\infty}}{\|\bx\|_{2,\infty}} \leq \rho_r
(\bba)
\end{equation}
which establishes (\ref{eq:lemi}). The proof of (\ref{eq:lem1}) is
similar:
\begin{eqnarray}
 \|\bba\bx\|_{2,1} & = & \sum_j \left\|\sum_{i}
\bba[j,i]\bx[i]\right\|_2  \nonumber \\
& \leq & \sum_j \sum_{i} \left\|\bba[j,i]\bx[i]\right\|_2
\nonumber \\
& \leq & \sum_i \|\bx[i]\|_2 \sum_j \rho(\bba[j,i]) \nonumber \\
& \leq  & \rho_c(\bba) \|\bx\|_{2,1} \label{eq:axrho}
\end{eqnarray}
from which the result follows. Finally, we have $\rho_{c}(\bba^{H})=\max_{r} \sum_{\ell}\rho(\bba^{H}[\ell,r])=\max_{r}\sum_{\ell}\rho(\bba[r,\ell])=\rho_{r}(\bba)$.

\section{Proof of Lemma~\ref{lemma:mnorm}}
\label{app:matrixnorm}

Nonnegativity and positivity follow immediately from the fact that the spectral norm is a matrix norm \cite[p.~295]{Horn85}.
Homogeneity follows by noting that
\begin{align}
\rho_{c}(\alpha\bba)&=\max_{r}\sum_{\ell}\rho(\alpha \bba[\ell,r]) \nonumber \\
&=\max_{r}\sum_{\ell}|\alpha|\rho(\bba[\ell,r]) \nonumber \\
&=|\alpha|\rho_{c}(\bba).
\end{align}
The triangle inequality is obtained as follows:
\begin{eqnarray*}
\rho_{c}(\bba+\bbb) & = & \max_{r}\sum_{\ell}\rho(\bba[\ell,r]+\bbb[\ell,r])\\
& \le & \max_{r} \left( \sum_{\ell}\rho(\bba[\ell,r])+ \sum_{\ell} \rho(\bbb[\ell,r])\right)\\
& \le & \max_{r}\sum_{\ell}\rho(\bba[\ell,r])+\max_{r}\sum_{\ell}\rho(\bbb[\ell,r])\\
& = & \rho_{c}(\bba)+\rho_{c}(\bbb)
\end{eqnarray*}
where the first inequality is a consequence of the spectral norm satisfying the triangle inequality.

Finally, to verify submultiplicativity, note that,
\begin{equation}
\label{eq:mxb} \rho_c(\bba\bbb) =\max_\ell \rho_c(\bba\bbb[\ell]).
\end{equation}
Therefore, if we prove that
\begin{equation}
\label{eq:nlemma}
 \rho_c(\bba\bbb[\ell]) \leq  \rho_c(\bba)\rho_c(\bbb[\ell])
\end{equation}
the result follows from (\ref{eq:mxb}) and the fact that
$\max_\ell \rho_c(\bbb[\ell])=\rho_c(\bbb)$.

To prove (\ref{eq:nlemma}), note that
\begin{align}
\label{eq:nlemma1}
 \rho_c(\bba\bbb[\ell]) &=\sum_{i} \rho\! \bl  \sum_j
 \bba[i,j]\bbb[j,\ell]\br \nonumber\\
 & \leq \sum_{i} \sum_j \rho\bl
 \bba[i,j]\bbb[j,\ell]\br \nonumber \\
 &\leq \sum_{i} \sum_j \rho( \bba[i,j])\rho(\bbb[j,\ell])
\end{align}
where we used the triangle inequality for, and the submultiplicativity of, the spectral norm.
Now, we have
\begin{equation}
\label{eq:nlemma2} \sum_i \rho( \bba[i,j]) \leq \max_{\ell} \sum_i
\rho( \bba[i,\ell])=\rho_c(\bba).
\end{equation}
Substituting into (\ref{eq:nlemma1}) yields
\begin{equation}
 \rho_c(\bba\bbb[\ell])
 \leq \rho_c(\bba)\sum_j
 \rho(\bbb[j,\ell])=\rho_c(\bba)\rho_c(\bbb[\ell])
\end{equation}
which completes the proof.

\section{Proof of Lemma~\ref{lemma:1inq}}
\label{app:1inq}

The proof of the statement $\|\bba \bv\|_{2,1}\,\le\,\rho_c(\bba) \|\bv\|_{2,1}$ follows directly from (\ref{eq:axrho}) by replacing
$\bba$ by an $L\,\times\,(kd)$ matrix and $\bx\,\in\,\CC^N$ by $\bv \, \in \, \CC^{kd}$ with $\|\bv[l]\|_{2}\,>\,0$, for all $\ell$.
If the $a_i=\sum_j \rho(\bba[j,i])$ are
not all equal, then the last inequality in (\ref{eq:axrho}) is strict. Since
$a_i=\rho_c(\bba \bbj_i)$ the result follows.


\bibliographystyle{IEEEtran}
\bibliography{masterbib}

\end{document}